\documentclass[11pt]{article}

\usepackage{amssymb,amsmath,amsthm,sectsty,url}
\usepackage[letterpaper,hmargin=1.0in,vmargin=1.0in]{geometry}
\usepackage{color}
\usepackage[boxed]{algorithm}
\usepackage{tikz}
\usepackage{graphicx}
\usepackage{nicefrac}
\usepackage{aliascnt}

\usepackage[utf8]{inputenc} 
\usepackage[T1]{fontenc}    
\usepackage{url}            
\usepackage{booktabs}       
\usepackage{amsfonts}       
\usepackage{nicefrac}       
\usepackage{microtype}      

\usepackage{times}
\usepackage{graphicx} 
\usepackage{subfigure} 

\usepackage{amssymb,amsmath,amsthm,sectsty,url}
\usepackage[pdftex,colorlinks,linkcolor=blue,citecolor=blue,filecolor=blue,urlcolor=blue]{hyperref}
\usepackage{cleveref}
\usepackage{color}

\newtheorem{theorem}{Theorem}
\crefname{theorem}{Theorem}{Theorems}

\newtheorem{definition}{Definition}
\crefname{definition}{Definition}{Definitions}

\newtheorem{lemma}{Lemma}
\crefname{lemma}{Lemma}{Lemmas}

\newtheorem{observation}{Observation}
\crefname{observation}{Observation}{Observations}

\newcommand{\norm}[1]{\lVert#1\rVert}

\newcommand{\R}{\mathbb R}

\title{Practical Data-Dependent Metric Compression\\
with Provable Guarantees\footnote{Code available at \url{https://github.com/talwagner/quadsketch}.}}

\author{
  Piotr Indyk\thanks{Authors ordered alphabetically.} \\
  MIT \\
  \texttt{indyk@mit.edu}
  \and
  Ilya Razenshteyn\footnotemark[2] \thanks{Supported by Simons Foundation Junior Fellowship.} \\
  Columbia University \\
  \texttt{ilyaraz@mit.edu}
  \and
  Tal Wagner\footnotemark[2] \\
  MIT \\
  \texttt{talw@mit.edu}
}

\begin{document}

\maketitle

\begin{abstract} 
 We introduce a new distance-preserving compact representation of multi-dimensional point-sets. Given $n$ points in a $d$-dimensional space where each coordinate is represented using $B$ bits (i.e., $dB$ bits per point), it produces  a representation of size $O( d \log(d B/\epsilon) + \log n)$ bits per point from which one can approximate the distances up to a factor of $1 \pm \epsilon$. Our algorithm almost matches the recent bound of~\cite{indyk2017near} while being much simpler. We compare our algorithm to Product Quantization (PQ)~\cite{jegou2011product}, a state of the art heuristic metric compression method. We evaluate both algorithms on several data sets: SIFT (used in \cite{jegou2011product}), MNIST~\cite{lecun1998mnist}, New York City taxi time series~\cite{guha2016robust} and a synthetic one-dimensional data set embedded in a high-dimensional space. 
With appropriately tuned parameters, our algorithm produces representations that are comparable to or better than those produced by PQ, while having provable guarantees on its performance.
 \end{abstract} 

\section{Introduction}

Compact distance-preserving representations of high-dimensional objects are very useful tools in data analysis and machine learning. They compress each data point in a data set using a small number of  bits while preserving the distances between the points up to a controllable accuracy. This makes it possible to run data analysis algorithms, such as similarity search, machine learning classifiers, etc,  on data sets of reduced size. The benefits of this approach include: (a) reduced running time (b) reduced storage and (c) reduced  communication cost (between machines, between CPU and RAM, between CPU and GPU, etc). These three factors make the computation more efficient overall, especially  on modern architectures where the communication cost is often the dominant factor in the running time, so fitting the data in a single processing unit is highly beneficial. Because of these benefits, various compact representations have been extensively studied over the last decade, for applications such as: speeding up similarity search \cite{broder1997resemblance, indyk1998approximate, kushilevitz2000efficient, torralba2008small, weiss2009spectral, jegou2011product, norouzi2012hamming, shrivastava2014densifying},  scalable learning algorithms~\cite{weinberger2009feature, li2011hashing}, streaming algorithms ~\cite{muthukrishnan2005data} and other tasks. For example, a recent paper~\cite{JohnsonDJ17} describes a similarity search software package based on one such method (Product Quantization (PQ))  that has been used to solve very large similarity search problems over billions of point on GPUs at Facebook. 


The methods for designing such representations can be classified into {\em data-dependent} and {\em data-oblivious}. The former analyze the whole data set in order to construct the point-set representation, while the latter apply a fixed procedure individually to each data point. A classic example of the data-oblivious approach is based on randomized dimensionality reduction  ~\cite{johnson1984extensions}, which states that any set of $n$ points in the Euclidean space of arbitrary dimension $D$ can be mapped into a space of dimension $d=O(\epsilon^{-2} \log n)$, such that the distances between all pairs of points are preserved up to a factor of $1\pm\epsilon$. This allows representing each point using $d (B+\log D)$ bits, where $B$ is the number of bits of precision in the coordinates of the original pointset.
\footnote{The bounds can be stated more generally in terms of the~\emph{aspect ratio} $\Phi$ of the point-set. See Section~\ref{s:formal} for the discussion.}
More efficient representations are possible if the goal is to preserve only the distances in a certain range. 
In particular, $O(\epsilon^{-2} \log n)$ {\em bits} are sufficient to distinguish between distances smaller than $1$ and greater than $1+\epsilon$, independently of the precision parameter~\cite{kushilevitz2000efficient} (see also \cite{raginsky2009locality} for kernel generalizations).
Even more efficient methods are known if the coordinates are binary~\cite{broder1997resemblance, li2011hashing, shrivastava2014densifying}.

Data-dependent methods compute the bit representations of points ``holistically", typically by solving a global optimization problem. Examples of this approach include Semantic Hashing~\cite{salakhutdinov2009semantic}, Spectral Hashing~\cite{weiss2009spectral}  or Product Quantization~\cite{jegou2011product} (see also the survey~\cite{wang2016learning}).  Although successful,  most of the results in this line of research are empirical in nature, and we are not aware of any worst-case accuracy vs.~compression tradeoff bounds for those methods along the lines of the aforementioned data oblivious approaches. 

A recent work~\cite{indyk2017near} shows that it is possible to combine the two approaches and obtain algorithms that  adapt to the data while providing  worst-case accuracy/compression tradeoffs. In particular, the latter paper shows how to construct representations of $d$-dimensional pointsets that preserve all distances up to a factor of $1\pm\epsilon$ while using only $O( (d + \log n) \log(1/\epsilon) + \log (Bn) )$ bits per point. Their algorithm uses hierarchical clustering in order to group close points together, and represents each point by a displacement vector from a near by point that has already been stored. The displacement vector is then appropriately rounded to reduce the representation size. Although theoretically interesting, that algorithm is rather complex and (to the best of our knowledge) has not been implemented. 

\paragraph{Our results.} The main contribution of this paper is QuadSketch (QS), a {\em simple} data-adaptive algorithm, which is both provable and practical. It represents each point using $O( d \log(d B/\epsilon) + \log n)$ bits, where (as before) we can set $d=O(\epsilon^{-2} \log n)$ using the Johnson-Lindenstrauss lemma. Our bound significantly improves over the ``vanilla'' $O(d B)$ bound (obtained by storing all $d$ coordinates to full precision), and comes close to bound of \cite{indyk2017near}. 
At the same time, the algorithm is quite simple and intuitive: it computes a $d$-dimensional quadtree\footnote{Traditionally, the term ``quadtree'' is used for the case of $d=2$, while its higher-dimensional variants are called `` hyperoctrees''~\cite{yau1983hierarchical}. However, for the sake of simplicity, in this paper we use the same term ``quadtree'' for any value of $d$. }  and appropriately prunes its edges and nodes.\footnote{We note that a similar idea (using kd-trees instead of quadtrees) has been earlier proposed in~\cite{arandjelovic2014extremely}. However, we are not aware of any provable space/distortion tradeoffs for the latter algorithm.}

We evaluate QuadSketch experimentally on both real and synthetic data sets: a SIFT feature data set from~\cite{jegou2011product}, MNIST~\cite{lecun1998mnist}, time series data reflecting taxi ridership in New York City~\cite{guha2016robust} and a synthetic data set (Diagonal) containing random points from a one-dimensional subspace (i.e., a line) embedded in a high-dimensional space. The data sets are quite diverse: SIFT and MNIST data sets are de-facto ``standard'' test cases for nearest neighbor search and distance preserving sketches, NYC taxi data was designed to contain anomalies and ``irrelevant'' dimensions, while Diagonal has extremely low intrinsic dimension. We compare our algorithms to Product Quantization (PQ)~\cite{jegou2011product}, a state of the art method for computing distance-preserving sketches, as well as a baseline simple uniform quantization method (Grid). The sketch length/accuracy tradeoffs for QS and PQ are comparable on SIFT and MNIST data, with PQ having higher accuracy for shorter sketches while QS having  better accuracy for longer sketches. On NYC taxi data, the accuracy of QS is higher over the whole range of sketch lengths . Finally, Diagonal exemplifies a situation where the low dimensionality of the data set hinders the performance of PQ, while QS naturally adapts to this data set. Overall, QS performs well on ``typical'' data sets, while its provable guarantees ensure robust performance in a wide range of scenarios. Both algorithms improve over the baseline quantization method. 




\section{Formal Statement of Results}
\label{s:formal}
\paragraph{Preliminaries.}
Let $X=\{x_1,\ldots,x_n\}\subset\R^d$ be a pointset in Euclidean space.
A compression scheme constructs from $X$ a bit representation referred to
as a~\emph{sketch}. Given the sketch, and without
access to the original pointset, one can~\emph{decompress} the sketch into
an approximate pointset $\tilde X=\{\tilde x_1,\ldots,\tilde x_n\}\subset\R^d$.
The goal is to minimize the size of the sketch, while approximately preserving
the geometric properties of the pointset, in particular the distances and near neighbors.

In the previous section we parameterized the sketch size in terms of the number
of points $n$, the dimension $d$, and the bits per coordinate $B$.
In fact, our results are more general, and can be stated in terms of the 
~\emph{aspect ratio} of the pointset, denoted
by $\Phi$ and defined as the ratio between the largest to smallest distance,
\[ \Phi = \frac{\max_{1\leq i<j \leq n}\norm{x_i-x_j}}{\min_{1\leq i<j \leq n}\norm{x_i-x_j}} . \]

Note that  $\log(\Phi) \leq \log d +B$, so our bounds, stated in terms of $\log \Phi$, immediately imply analogous bounds in terms of $B$.

We will use $[n]$ to denote $\{1,\ldots,n\}$, and $\tilde O(f)$ to suppress
polylogarithmic factors in $f$.

\paragraph{QuadSketch.}
Our compression algorithm, described in detail in Section~\ref{s:compression},  is based on a randomized variant
of a quadtree followed by a pruning step.
In its simplest variant, the trade-off between the sketch size and compression quality
is governed by a single parameter $\Lambda$.
Specifically, $\Lambda$  controls the pruning step, in which
the algorithm identifies ``non-important'' bits among those stored in the
quadtree (i.e.~bits whose omission would have little effect on the approximation
quality), and removes them from the sketch.
Higher values of $\Lambda$ result in sketches that are longer but have better approximation quality.

%
%

\paragraph{Approximate nearest neighbors.}
Our main theorem provides the following guarantees for the  basic variant of QuadSketch: for each point, the distances from that point to all other points are preserved up to a factor of $1 \pm \epsilon$ with a constant probability. 

\begin{theorem}\label{thm:ann}
Given $\epsilon,\delta>0$, let $\Lambda=O(\log(d\log\Phi/\epsilon\delta))$ and $L=\log\Phi+\Lambda$.
QuadSketch runs in time $\tilde O(ndL)$ and produces
a sketch of size $O(nd\Lambda+n\log n)$ bits, with the following guarantee:
For every $i\in[n]$,
\[ \Pr\left[\forall_{j\in[n]}\norm{\tilde x_i-\tilde x_j} = (1\pm\epsilon)\norm{x_i-x_j} \right] \geq 1-\delta. \]
In particular, with probability $1-\delta$,
if $\tilde x_{i^*}$ is the nearest neighbor of $\tilde x_i$ in $\tilde X$,
then $x_{i^*}$ is a $(1+\epsilon)$-approximate nearest neighbor of $x_i$ in $X$.
\end{theorem}

Note that the theorem allows us to compress the input point-set into a sketch and then decompress it back into a point-set which can be fed to a black box similarity search algorithm.  Alternatively, one can decompress only specific points and approximate the distance between them. 

%
For example, if $d=O(\epsilon^{-2}\log n)$ and $\Phi$
is polynomially bounded in $n$,
then~\cref{thm:ann} uses $\Lambda=O(\log\log n + \log(1/\epsilon))$ bits per coordinate 
to preserve $(1+\epsilon)$-approximate nearest neighbors.

The full version of QuadSketch, described in Section~\ref{s:compression}, allows extra fine-tuning by exposing additional parameters of the algorithm. The guarantees for  the full version are summarized by Theorem~\ref{thm:block-ann} in Section~\ref{s:compression}.

\paragraph{Maximum distortion.}
We also show that a recursive application of QuadSketch makes it possible to approximately preserve the distances between {\em all} pairs of points. 
This is the setting considered in~\cite{indyk2017near}. 
(In contrast, \cref{thm:ann} preserves the distances from any single point.)

\begin{theorem}\label{thm:max_distortion}
Given $\epsilon>0$, let $\Lambda=O(\log(d\log\Phi/\epsilon))$ and $L=\log\Phi+\Lambda$.
There is a randomized algorithm that runs in time 
$\tilde O(ndL)$ and produces a sketch of size
$O(nd\Lambda+n\log n)$ bits, 
such that with high probability, every distance $\norm{x_i-x_j}$
can be recovered from the sketch up to distortion $1\pm\epsilon$.
\end{theorem}

\Cref{thm:max_distortion} has smaller sketch size than that provided by the ``vanilla'' bound, and only slightly larger than that in~\cite{indyk2017near}.
For example, for $d=O(\epsilon^{-2}\log n)$ and $\Phi=\mathrm{poly}(n)$,
it improves over the ``vanilla'' bound by a factor of $O(\log n/\log\log n)$ 
and is lossier than the bound of~\cite{indyk2017near} by  a factor of $O(\log\log n)$.
However, compared to the latter, our construction time is nearly linear in $n$.
The comparison is summarized in~\Cref{tbl:sketches_related_work}.

We remark that~\cref{thm:max_distortion} does not let us recover an approximate embedding 
of the pointset, $\tilde x_1,\ldots,\tilde x_n$, as~\cref{thm:ann} does. Instead, the sketch
functions as an oracle that accepts queries of the form $(i,j)$ and return an approximation for
the distance $\norm{x_i-x_j}$.

\begin{table}[t]
\caption{Comparison of Euclidean metric sketches with maximum distortion $1\pm\epsilon$,
for $d=O(\epsilon^{-2}\log n)$ and $\log\Phi=O(\log n)$.}
\label{tbl:sketches_related_work}
\vskip 0.15in
\begin{center}
\begin{small}
\begin{tabular}{lll}
\hline
\textsc{Reference} & \textsc{Bits per point} & \textsc{Construction time} \\
\hline
``Vanilla'' bound   & $O(\epsilon^{-2}\log^2n)$ & -- \rule{0pt}{3ex} \\
Algorithm of \cite{indyk2017near} &  $O(\epsilon^{-2}\log n\ \log(1/\epsilon))$ &  $\tilde O(n^{1+\alpha}+\epsilon^{-2}n)$ for $\alpha\in(0,1]$ \rule{0pt}{3ex} \\
\Cref{thm:max_distortion}   & $O(\epsilon^{-2}\log n\ (\log\log n\ + \log(1/\epsilon)))$ &  $\tilde O(\epsilon^{-2}n)$ \rule{0pt}{3ex} \\
\hline
\end{tabular}
\end{small}
\end{center}
\end{table}

\section{The Compression Scheme}
\label{s:compression}
The sketching algorithm takes as input the pointset $X$, and two parameters $L$ and $\Lambda$ that
control the amount of compression.

\paragraph{Step 1: Randomly shifted grid.}
The algorithm starts by imposing a randomly shifted axis-parallel grid on the points.
We first enclose the whole pointset in an axis-parallel hypercube $H$.
Let $\Delta'=\max_{i\in[n]}\norm{x_1-x_i}$, and $\Delta=2^{\lceil\log\Delta'\rceil}$.
Set up $H$ to be centered at $x_1$ with side length $4\Delta$.
Now choose $\sigma_1,\ldots,\sigma_d\in[-\Delta,\Delta]$ independently and uniformly
at random, and shift $H$ in each coordinate $j$ by $\sigma_j$. By the choice of side length
$4\Delta$, one can see that $H$ after the shift still contains the whole pointset.
For every integer $\ell$ such that $-\infty<\ell\leq\log(4\Delta)$, let $G_\ell$ denote the
axis-parallel grid with cell side $2^\ell$ which is aligned with $H$.

Note that this step can be often eliminated in practice without affecting the empirical performance of the algorithm, but it is necessary in order to achieve guarantees for {\em arbitrary} pointsets. 

\paragraph{Step 2: Quadtree construction.}
The $2^d$-ary quadtree on the nested grids $G_\ell$ is naturally defined by associating every
grid cell $c$ in $G_\ell$ with the tree node at level $\ell$, such that its children are the $2^d$ grid cells in $G_{\ell-1}$ which are contained in $c$. The edge connecting a node $v$ to a child $v'$ is labeled with a bitstring of length $d$ defined as follows:  the $j^{th}$ bit is $0$ if $v'$ coincides with the bottom half of $v$ along coordinate $j$, and $1$ if $v'$ coincides with the upper half along that coordinate.

In order to construct the tree, we start with $H$ as the root, and bucket the points contained in
it into the $2^d$ children cells. We only add child nodes for cells that contain at least one point of
$X$. Then we continue by recursing on the child nodes. The quadtree construction is finished after $L$ levels.
We denote the resulting edge-labeled tree by $T^*$.
A construction for $L=2$ is illustrated in~\Cref{fig:quadtree}.

\begin{figure}[ht]
\vskip 0.2in
\begin{center}
\centerline{\includegraphics[scale=0.6]{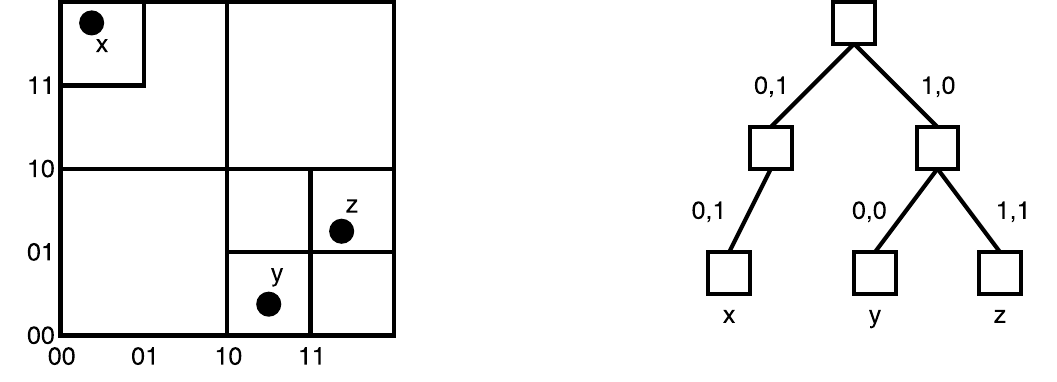}}
\caption{Quadtree construction for points $x,y,z$. The $x$ and $y$ coordinates are written as binary numbers.}
\label{fig:quadtree}
\end{center}
\vskip -0.2in
\end{figure} 

We define the~\emph{level} of a tree node with side length $2^\ell$ to be $\ell$ (note that $\ell$ can be negative).
The~\emph{degree} of a node in $T^*$ is its number of children.
Since all leaves are located at the bottom level, each point $x_i\in X$ is
contained in exactly one leaf, which we henceforth denote by $v_i$.


\paragraph{Step 3: Pruning.}
Consider a downward path $u_0,u_1,\ldots,u_k$ in $T^*$,
such that $u_1,\ldots,u_{k-1}$ are nodes with degree $1$,
and $u_0,u_k$ are nodes with degree other than $1$ ($u_k$ may be a leaf).
For every such path in $T^*$, if $k>\Lambda+1$, we remove the nodes $u_{\Lambda+1},\ldots,u_{k-1}$
from $T^*$ with all their adjacent edges (and edge labels).
Instead we connect $u_k$ directly to $u_{\Lambda}$ as its child. 
We refer to that edge as the~\emph{long edge}, and 
 label it with the length of the path it replaces ($k-\Lambda$).
The original edges from $T^*$ are called~\emph{short edges}.
At the end of the pruning step, we denote the resulting tree by $T$.

\paragraph{The sketch.}
For each point $x_i\in X$ the sketch stores the index of the leaf $v_i$ that contains it.
In addition it stores the structure of the tree $T$, encoded using the Eulerian Tour Technique\footnote{See e.g., 
https://en.wikipedia.org/wiki/Euler\_tour\_technique.}. 
Specifically, starting at the root, we traverse $T$ in the Depth First Search (DFS) order. 
In each step, DFS either explores the child of the current node (downward step), or returns to the parent node (upward step). We
encode a downward step by $0$ and an upward step by $1$. With each downward step we also store
 the label of the traversed edge (a length-$d$ bitstring for a short edge or the edge length
for a long edge, and an additional bit marking if the edge is short or long).

\paragraph{Decompression.}
Recovering $\tilde x_i$ from the sketch is done simply by following the downward path
from the root of $T$ to the associated leaf $v_i$, collecting the edge labels of the short edges,
and placing zeros instead of the missing bits of the long edges. The collected bits then correspond
to the binary expansion of the coordinates of $\tilde x_i$.

More formally, for every node $u$ (not necessarily a leaf) we define $c(u)\in\R^d$ as follows:
For $j\in\{1,\ldots,d\}$, concatenate
the $j^{th}$ bit of every short edge label traversed along the downward path from the root to $u$.
When traversing a long edge labeled with length $k$, concatenate $k$ zeros.\footnote{This is the ``lossy'' step in our sketching method: the original bits could be arbitrary, but they are replaced with zeros.}
Then, place a binary floating point in the resulting bitstring, 
after the bit corresponding to level $0$.
(Recall that the levels in $T$ are defined by the grid cell side lengths, and $T$
might not have any nodes in level $0$; in this case we need to pad with $0$'s either on
the right or on the left until we have a $0$ bit in the location corresponding to level $0$.)
The resulting binary string is the binary expansion of the $j^{th}$ coordinate of $c(u)$.
Now $\tilde x_i$ is defined to be $c(v_i)$.


\paragraph{Block QuadSketch.}We can further modify QuadSketch in a manner similar to Product Quantization~\cite{jegou2011product}. Specifically,  we partition the $d$ dimensions into $m$ blocks $B_1 \ldots B_m$ of size $d/m$  each, and apply QuadSketch separately to each block. More formally, for each $B_i$, we apply QuadSketch to the pointset $(x_1)_{B_i}  \ldots (x_n)_{B_i}$,  where $x_B$ denotes the $m/d$-dimensional vector obtained by projecting $x$ on the dimensions in $B$. 

The following statement is an immediate corollary of Theorem~\ref{thm:ann}.


\begin{theorem}\label{thm:block-ann}
Given $\epsilon,\delta>0$, and $m$ dividing $d$,  set the pruning parameter $\Lambda$ to $O(\log(d\log\Phi/\epsilon\delta))$ and  the number of levels $L$ to $\log\Phi+\Lambda$.
The $m$-block variant of QuadSketch runs in time $\tilde O(ndL)$ and produces
a sketch of size $O(nd\Lambda+nm\log n)$ bits, with the following guarantee:
For every $i\in[n]$,
\[ \Pr\left[\forall_{j\in[n]}\norm{\tilde x_i-\tilde x_j} = (1\pm\epsilon)\norm{x_i-x_j} \right] \geq 1-m\delta. \]
\end{theorem}

It can be seen that increasing the number of blocks $m$ up to a certain threshold ( $d \Lambda/\log n$ ) does not affect the asymptotic bound on the sketch size. Although we cannot prove that varying $m$ allows to  {\em improve} the accuracy of the sketch, this seems to be the case empirically, as demonstrated in the experimental section.

\section{Experiments}

We evaluate QuadSketch experimentally and compare its performance
to Product Quantization (PQ)~\cite{jegou2011product}, a state-of-the-art compression
scheme for approximate nearest neighbors,
and to a baseline of uniform scalar quantization, which we refer to as Grid.
For each dimension of the dataset, Grid places $k$ equally spaced
landmark scalars on the interval between the minimum and the maximum values along that dimension,
and rounds each coordinate to the nearest landmark.

All three algorithms work by partitioning the data dimensions into blocks,
and performing a quantization step in each block independently of the other ones.
QuadSketch and PQ take the number of blocks as a parameter, and Grid uses blocks of size $1$.
The quantization step is the basic algorithm described in~\Cref{s:compression} for QuadSketch,
$k$-means for PQ, and uniform scalar quantization for Grid.

We test the algorithms on four datasets:
The SIFT data used in~\cite{jegou2011product},
MNIST~\cite{lecun1998mnist} (with all vectors normalized to $1$), 
NYC Taxi ridership data~\cite{guha2016robust}, 
and a synthetic dataset called Diagonal,
consisting of random points on a line embedded in a high-dimensional space.
The properties of the datasets are summarized in~\Cref{tbl:datasets}.
Note that we were not able to compute the exact diameters for MNIST and SIFT, hence we only report estimates for $\Phi$ for these data sets, obtained via random sampling. 

The Diagonal dataset consists of $10,000$ points of the form $(x,x,\ldots,x)$,
where $x$ is chosen independently and uniformly at random from the interval $[0..40000]$.
This yields a dataset with a very large aspect ratio $\Phi$,
and on which partitioning into blocks is not expected to be beneficial since all coordinates are maximally correlated. 

For SIFT and MNIST we use the standard query set provided with each dataset.
For Taxi and Diagonal we use $500$ queries chosen at random from each dataset.
For the sake of consistency, for all data sets,  we apply the same quantization process jointly to both the point set and the query set, for both PQ and QS. We note, however, that both algorithms can be run on ``out of sample'' queries.

\begin{table}[t]
\caption{Datasets used in our empirical evaluation. The aspect ratio of SIFT and MNIST is estimated on a random sample.}
\label{tbl:datasets}
\vskip 0.15in
\begin{center}
\begin{small}
\begin{tabular}{lccc}
\hline
Dataset & Points & Dimension & Aspect ratio ($\Phi$)  \\
\hline
SIFT & $1,000,000$ & $128$ & $\geq 83.2$ \rule{0pt}{2ex} \\
MNIST & $60,000$ & $784$ & $\geq 9.2$ \rule{0pt}{2ex} \\
NYC Taxi & $8,874$ & $48$ & $49.5$ \rule{0pt}{2ex} \\
Diagonal (synthetic) & $10,000$ & $128$ & $20,478,740.2$ \rule{0pt}{2ex} \\
\hline
\end{tabular}
\end{small}
\end{center}
\end{table}

For each dataset, we enumerate the number of blocks over all divisors of the dimension $d$. 
For QuadSketch, $L$ ranges in $2,\ldots,20$, and $\Lambda$ ranges in $1,\ldots,L-1$.
For PQ, the number of $k$-means landmarks per block ranges in $2^5,2^6,\ldots,2^{12}$.
For both algorithms we include the results for all combinations of the parameters, and plot the envelope of the best performing combinations.

We report two measures of performance for each dataset: (a) the {\em accuracy}, defined as the fraction of queries for which the sketch returns the true nearest neighbor,
and (b) the {\em average distortion}, defined as the ratio between the (true) distances from the query
to the reported near neighbor and to the true nearest neighbor. 
The sketch size is measured in bits per coordinate.
The results appear in~\Cref{fig:sift,fig:mnist,fig:taxi,fig:diagonal}.
Note that the vertical coordinate in the distortion plots corresponds to the value of $\epsilon$, not $1+\epsilon$.

For SIFT, we also include a comparison with Cartesian k-Means (CKM)~\cite{norouzi2013cartesian}, in~\Cref{fig:ckm}.

\begin{figure}[p]
\begin{center}
\includegraphics[width=0.45\textwidth]{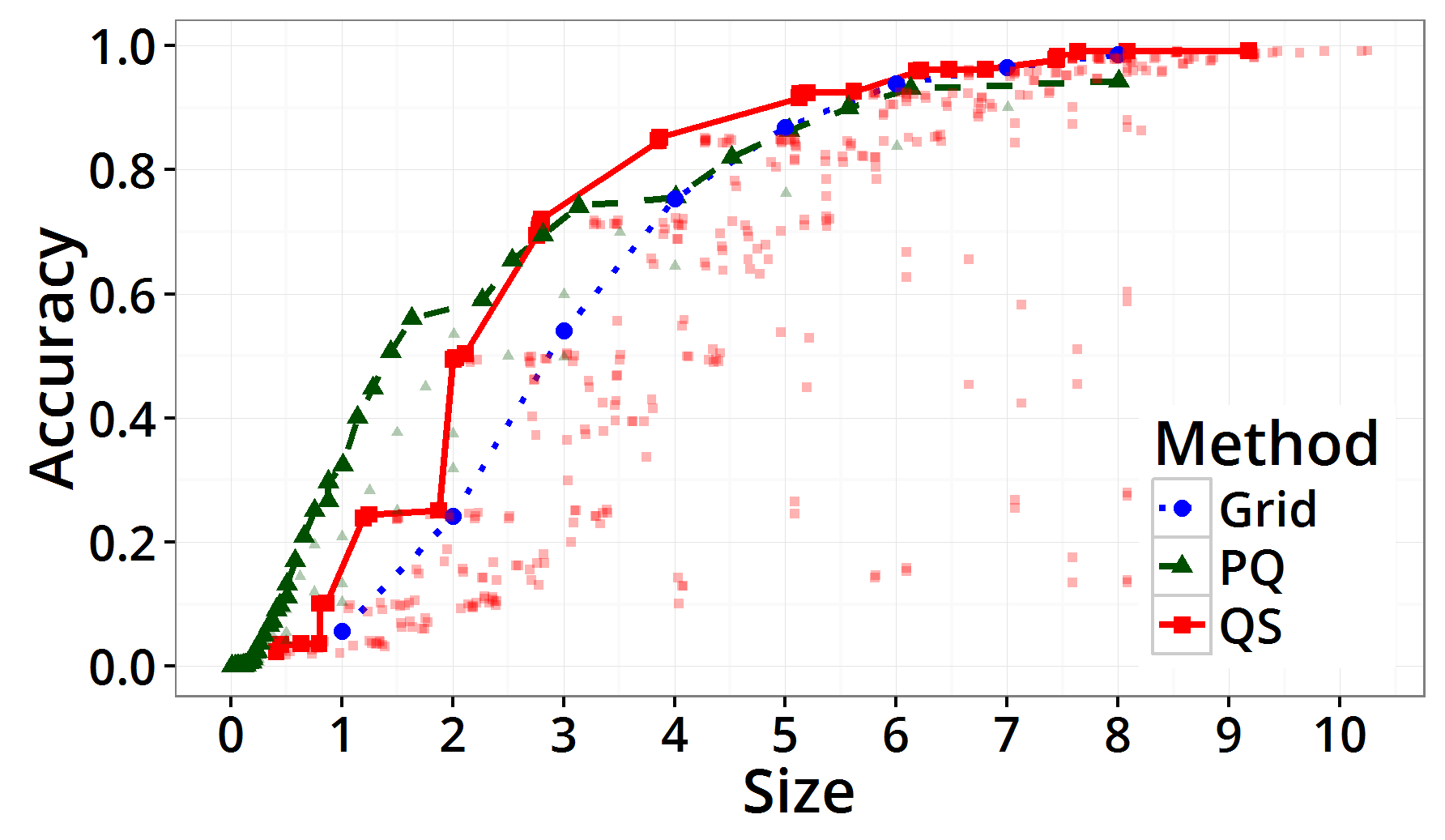}%
\includegraphics[width=0.45\textwidth]{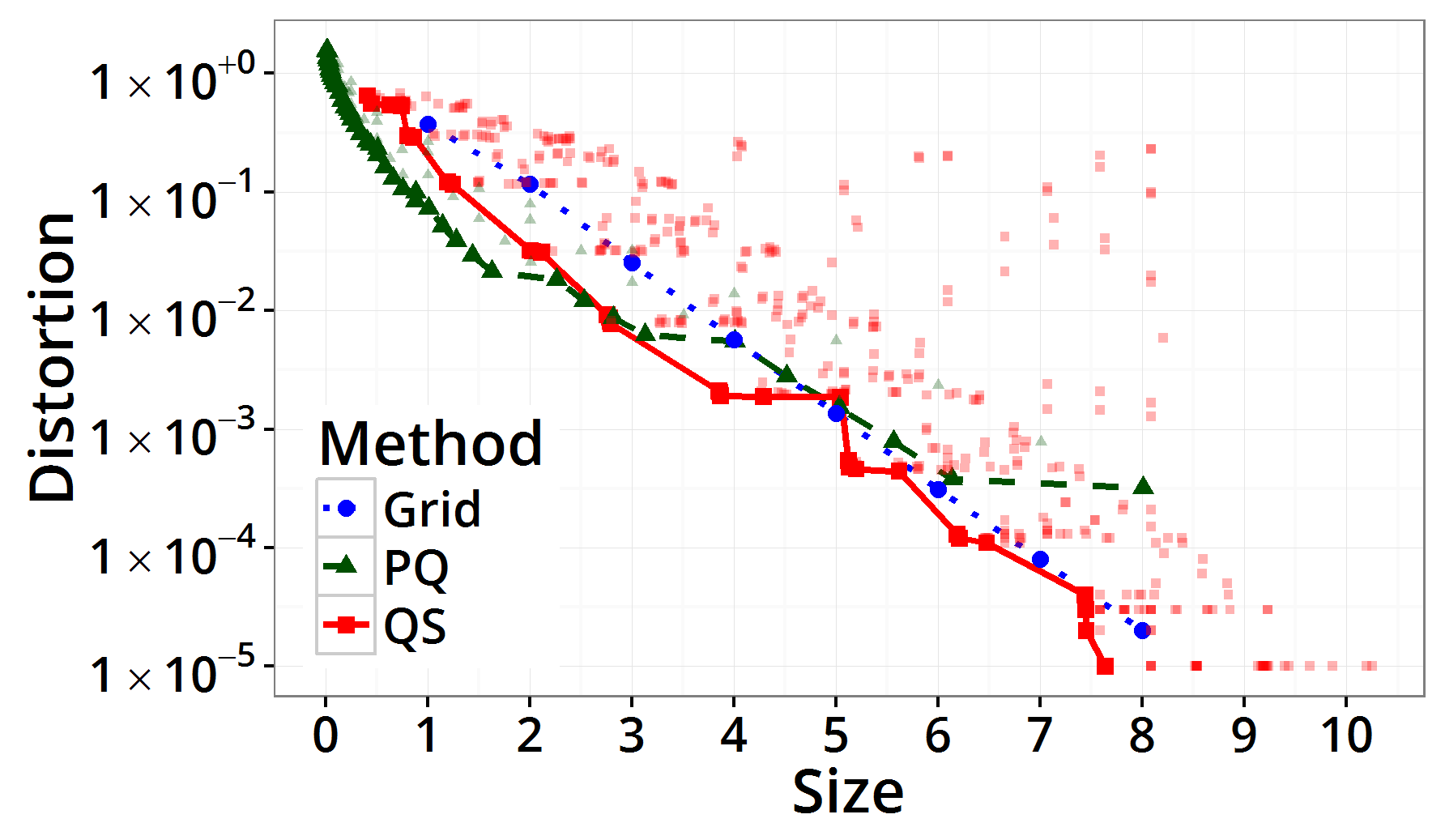}
\caption{Results for the SIFT dataset.}
\label{fig:sift}
\end{center}

\begin{center}
\includegraphics[width=0.45\textwidth]{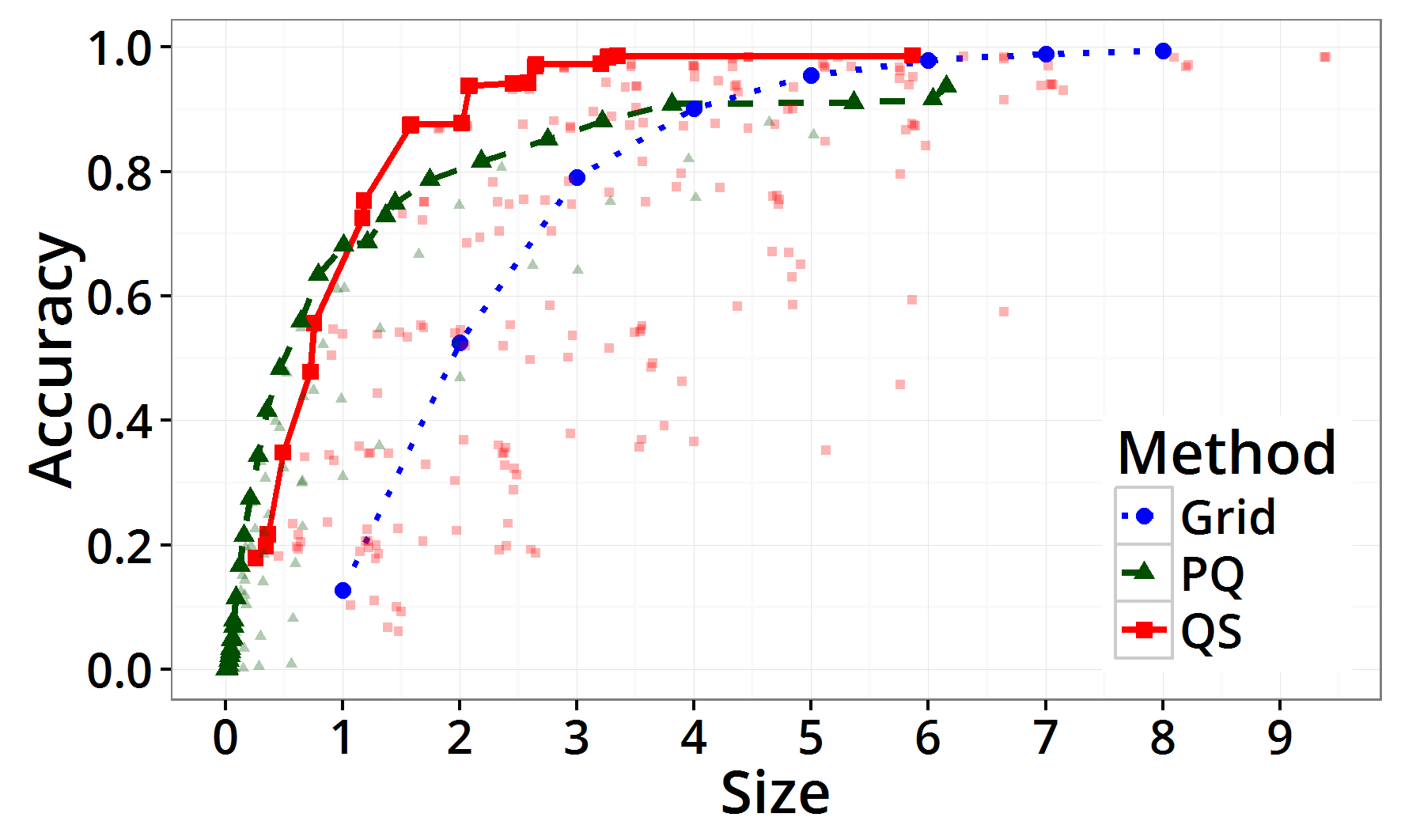}%
\includegraphics[width=0.45\textwidth]{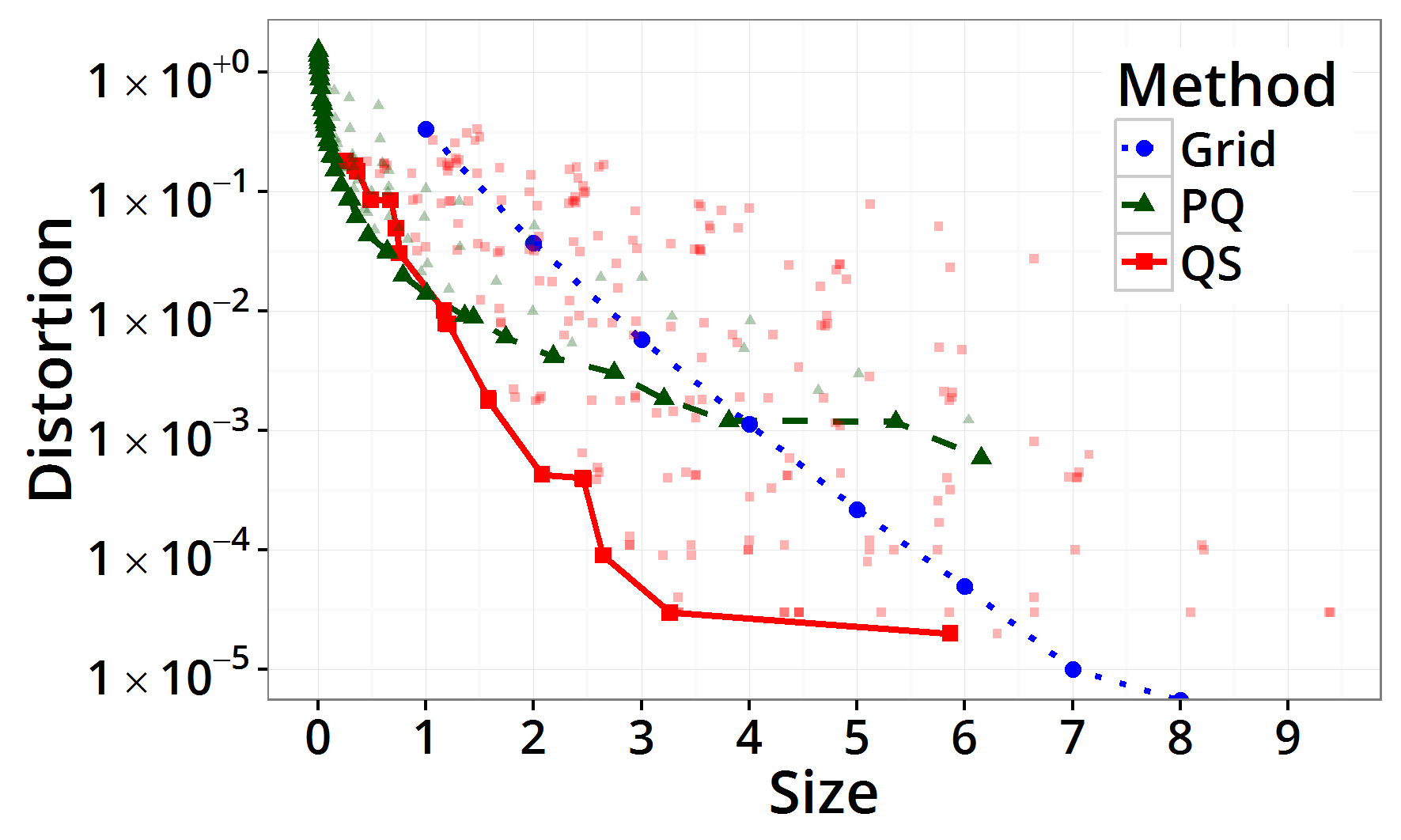}
\caption{Results for the MNIST dataset.}
\label{fig:mnist}
\end{center}

\begin{center}
\includegraphics[width=0.45\textwidth]{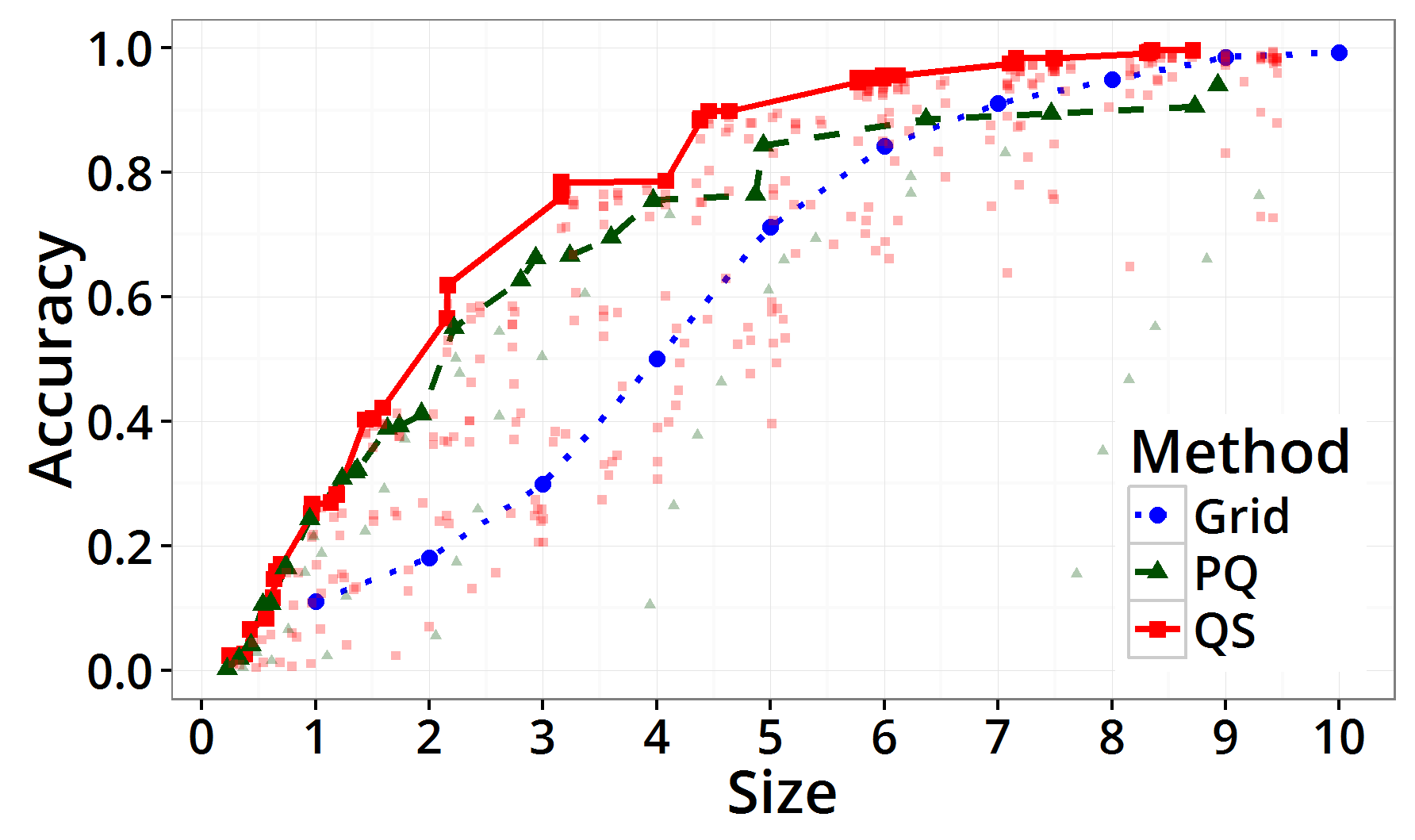}%
\includegraphics[width=0.45\textwidth]{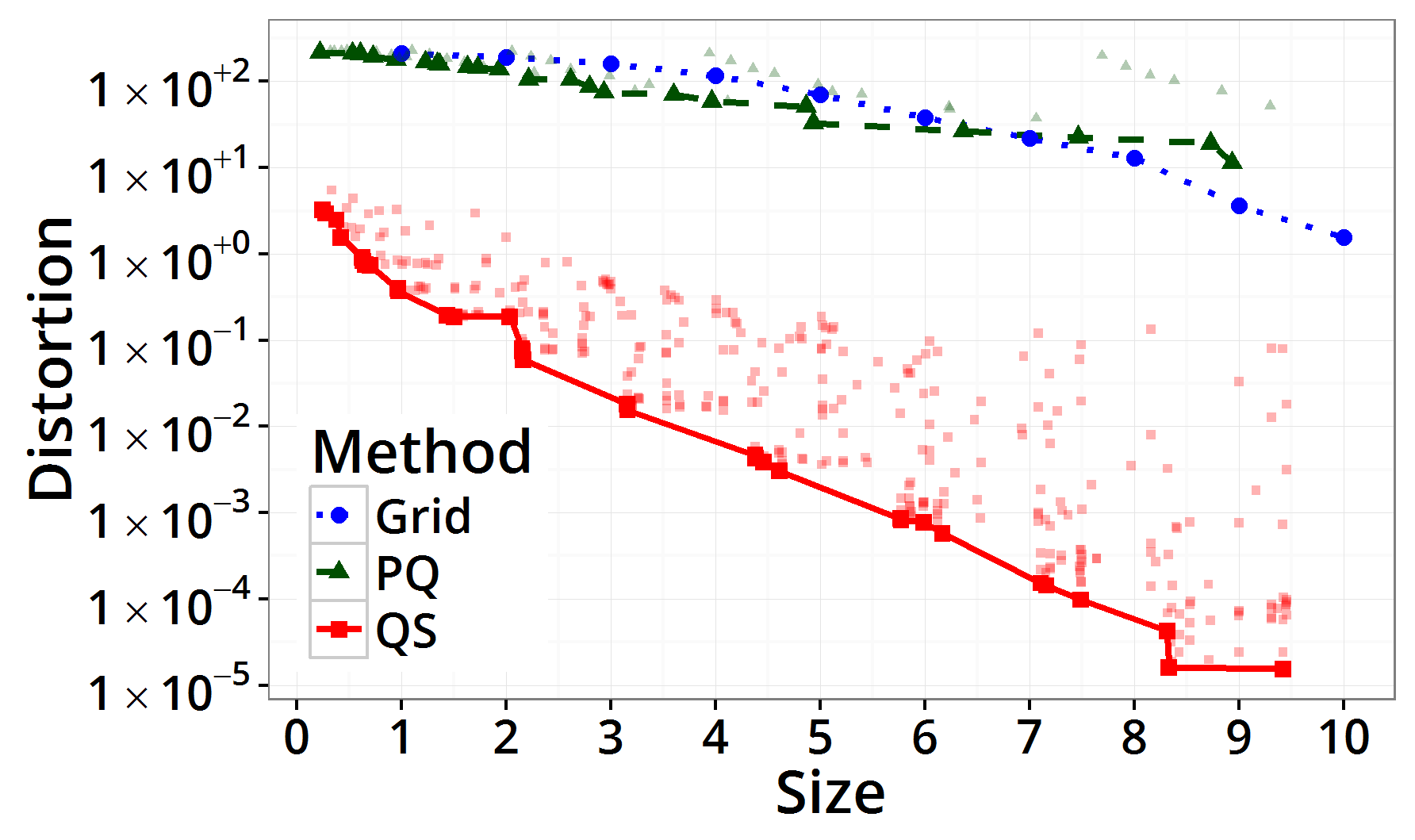}
\caption{Results for the Taxi dataset.}
\label{fig:taxi}
\end{center}

\begin{center}
\includegraphics[width=0.45\textwidth]{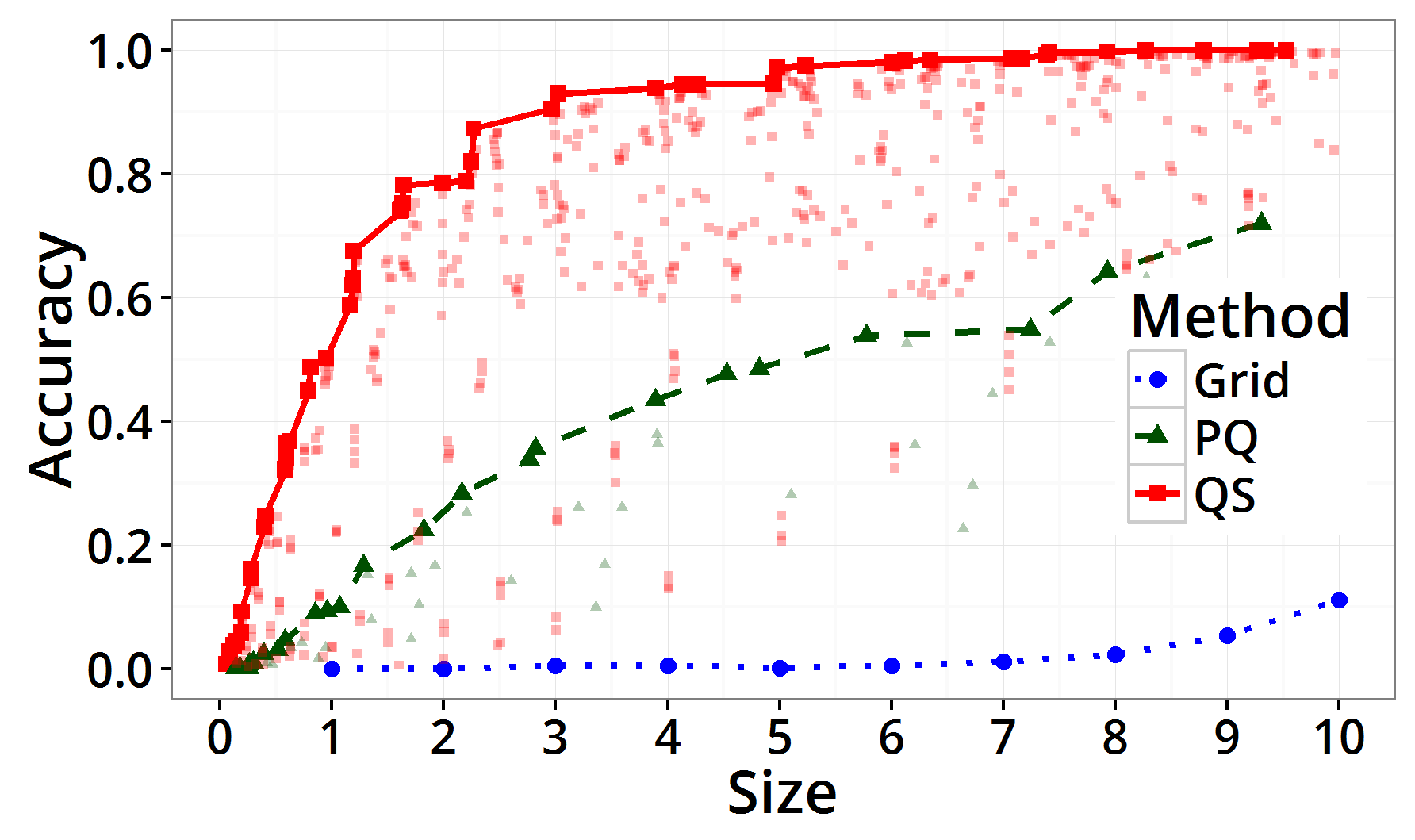}%
\includegraphics[width=0.45\textwidth]{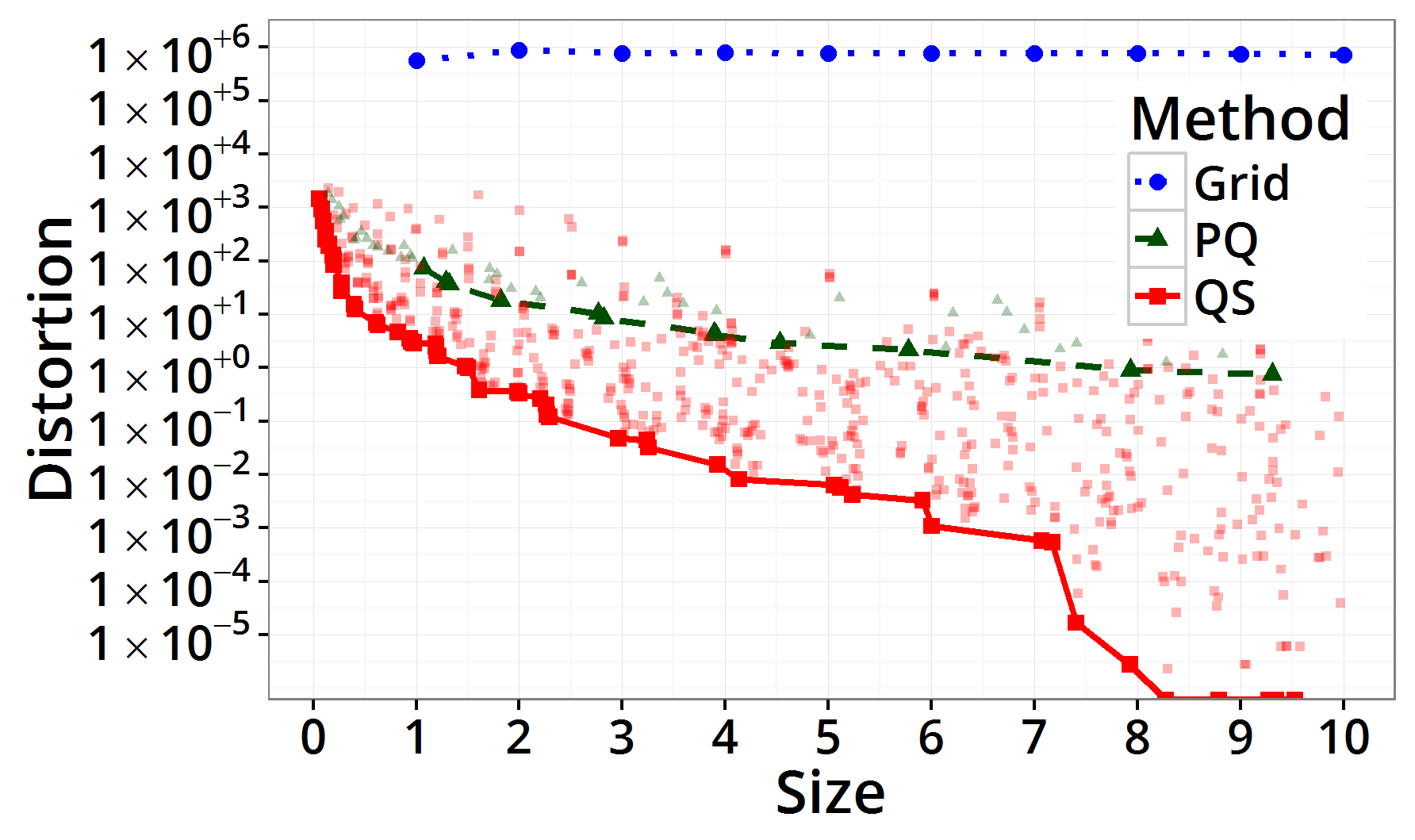}
\caption{Results for the Diagonal dataset.}
\label{fig:diagonal}
\end{center}
\end{figure}

\begin{figure}[h]
\begin{center}
\includegraphics[width=0.45\textwidth]{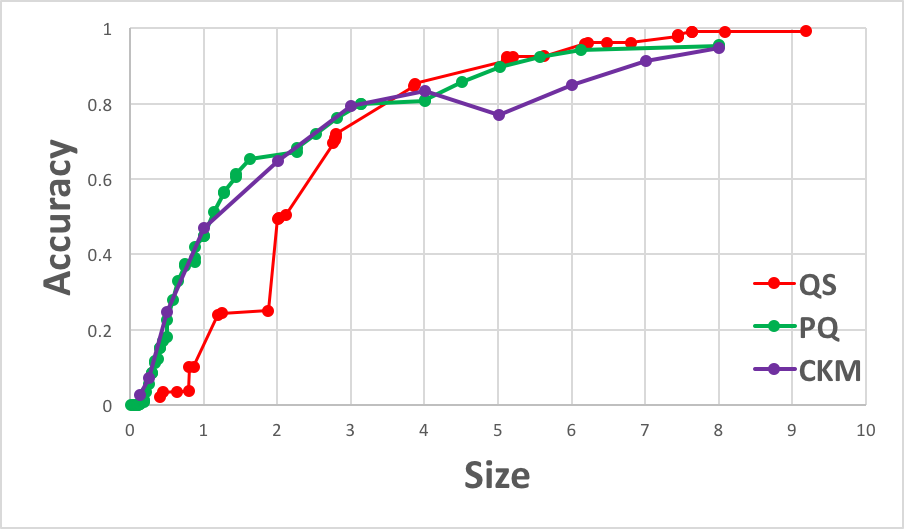}
\caption{Additional results for the SIFT dataset.}
\label{fig:ckm}
\end{center}
\end{figure}

\subsection{QuadSketch Parameter Setting}
We plot how the different parameters of QuadSketch effect its performance.
Recall that $L$ determines the number of levels in the quadtree prior to the pruning step,
and $\Lambda$ controls the amount of pruning.
By construction, the higher we set these parameters, the larger the sketch will be and with better accuracy.
The empirical tradeoff for the SIFT dataset is plotted in~\Cref{fig:depth_lambda}.

\begin{figure}[h]
\vskip 0.2in
\begin{center}
\includegraphics[width=0.45\textwidth]{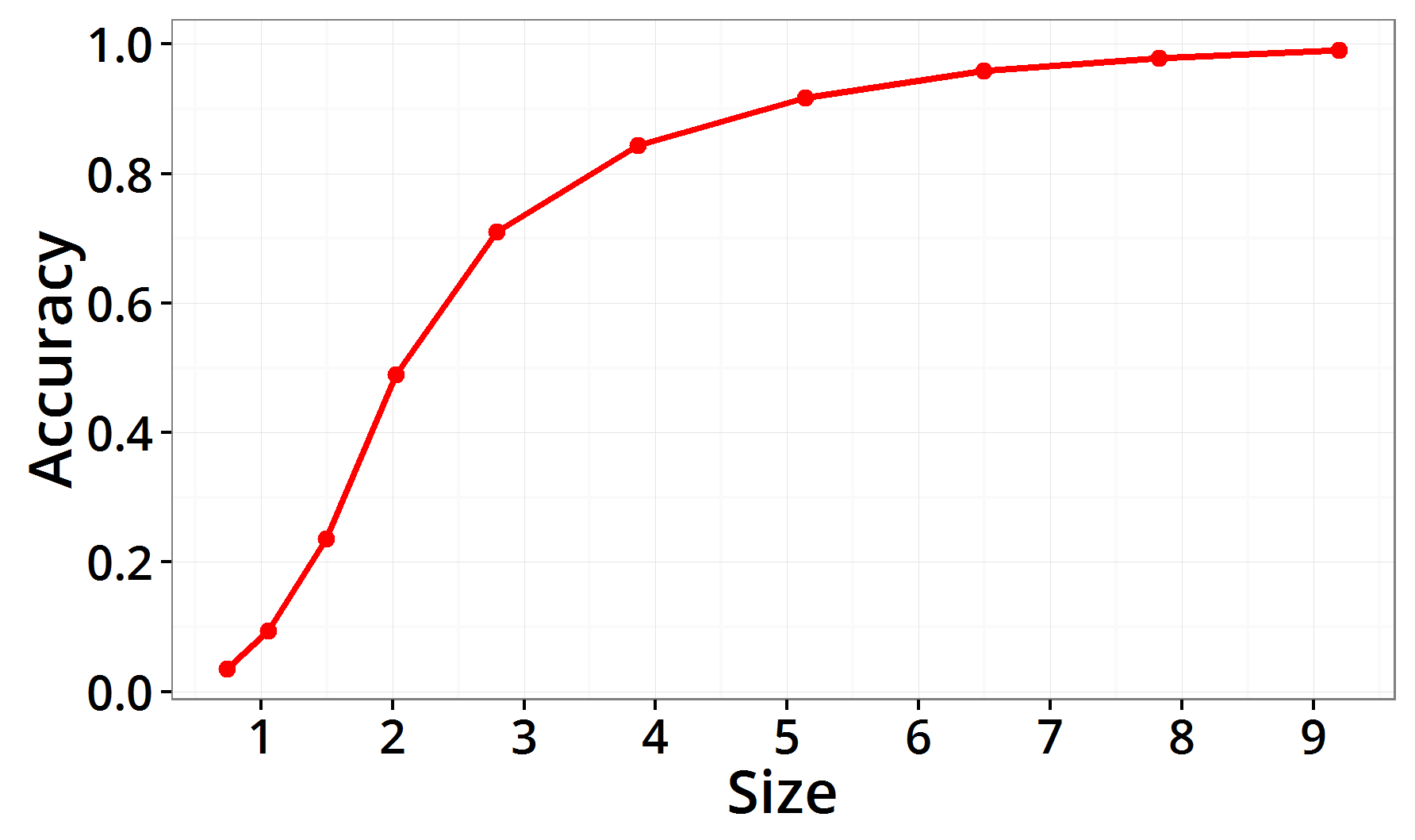}%
\includegraphics[width=0.45\textwidth]{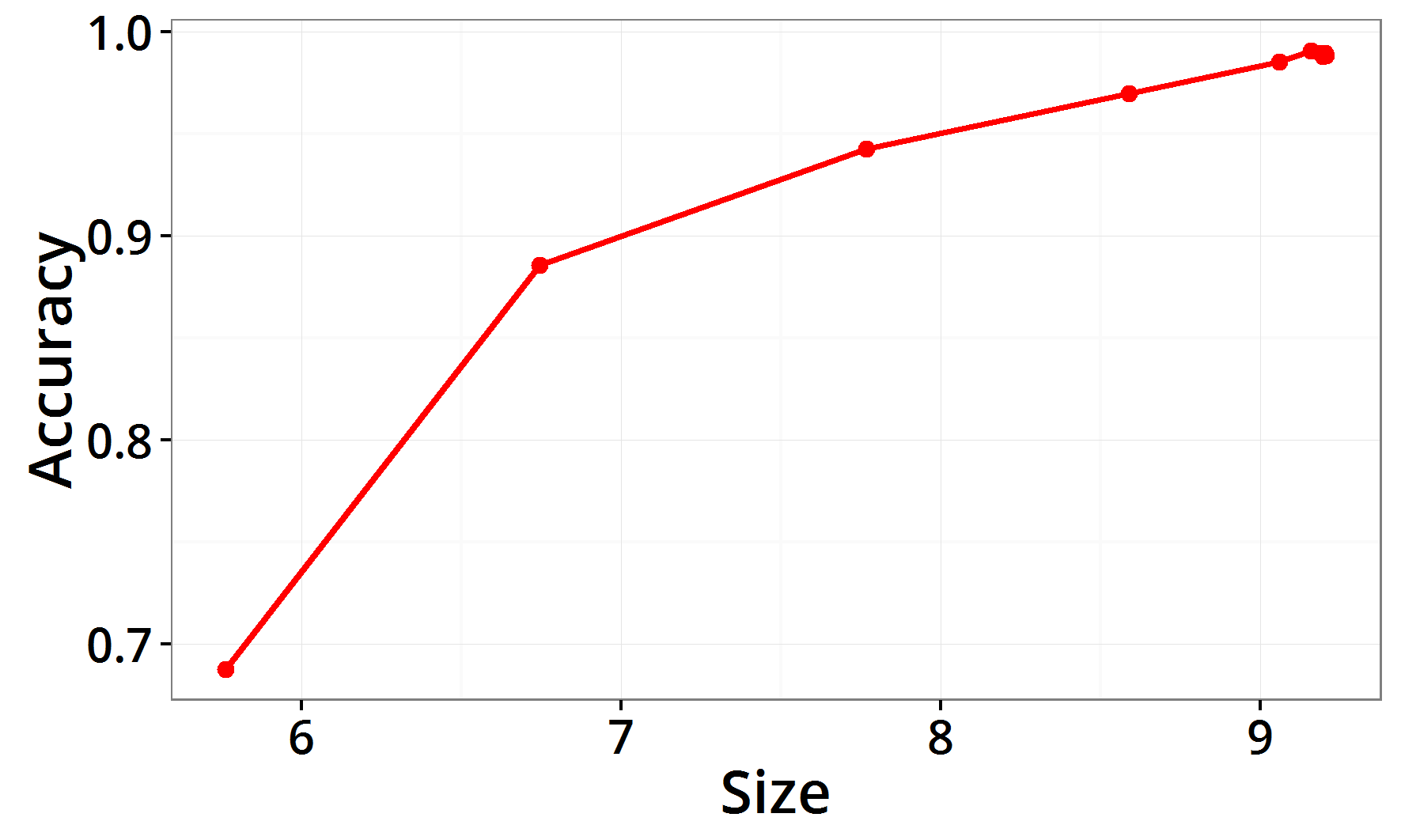}
\caption{On the left, $L$ varies from $2$ to $11$ for a fixed setting of $16$ blocks and $\Lambda=L-1$ (no pruning).
On the right, $\Lambda$ varies from $1$ to $9$ for a fixed setting of $16$ blocks and $L=10$. Increasing $\Lambda$ beyond $6$ does not have further effect on the resulting sketch.}

\label{fig:depth_lambda}
\end{center}
\vskip -0.2in
\end{figure}

The optimal setting for the number of blocks is not monotone,
and generally depends on the specific dataset.
It was noted in~\cite{jegou2011product} that on SIFT data an intermediate
number of blocks gives the best results, and this is confirmed by our experiments.
\Cref{tbl:blocks} lists the performance on the SIFT dataset for a varying number of blocks,
for a fixed setting of $L=6$ and $\Lambda=5$.
It shows that the sketch quality remains essentially the same, while the size varies significantly,
with the optimal size attained at $16$ blocks.

\begin{table}[h]
\vskip 0.15in
\begin{center}
\begin{small}
\begin{tabular}{cccc}
\hline
\# Blocks & Bits per coordinate & Accuracy & Average distortion  \\
\hline
$1$ & $5.17$ & $0.719$ & $1.0077$ \rule{0pt}{2ex} \\
$2$ & $4.523$ & $0.717$ & $1.0076$ \rule{0pt}{2ex} \\
$4$ & $4.02$ & $0.722$ & $1.0079$ \rule{0pt}{2ex} \\
$8$ & $3.272$ & $0.712$ & $1.0079$ \rule{0pt}{2ex} \\
$\mathbf{16}$ & $\mathbf{2.795}$ & $0.712$  & $1.008$ \rule{0pt}{2ex} \\
$32$ & $3.474$ & $0.712$  & $1.0082$ \rule{0pt}{2ex} \\
$64$ & $4.032$ & $0.713$  & $1.0081$ \rule{0pt}{2ex} \\
$128$ & $4.079$ & $0.72$  & $1.0078$ \rule{0pt}{2ex} \\
\hline
\end{tabular}
\end{small}
\end{center}
\caption{QuadSketch accuracy on SIFT data by number of blocks, with $L=6$ and $\Lambda=5$.}
\label{tbl:blocks}
\end{table}

\section{Proofs}
\label{s:proofs}
In this section we prove~\cref{thm:ann,thm:max_distortion}.
Recall that we have a pointset $x_1,\ldots,x_n\in\R^d$ with aspect ratio $\Phi$,
and given error parameters $\epsilon,\delta>0$.
For the remainder of the section we fix the setting
\[ \Lambda = \log(16d^{1.5}\log\Phi/(\epsilon\delta)). \]

Recall that in~\Cref{s:compression} we let $G_\ell$ denote the grid with side length $2^\ell$ for every integer $\ell$.
Our analysis is based on the observation that randomly shifting the grids,
which is used as the first step of our algorithm,
induces a \emph{padded decomposition}~\cite{bartal1996probabilistic} of the pointset.
We now define this formally.

\begin{definition}[padded point]\label{def:padded}
We say that a point $x_i$ is~\emph{$(\epsilon,\Lambda,\ell)$-padded}, if
the grid cell in $G_{\ell}$ that contains $x_i$
also contains the ball of radius $\rho(\ell)$ centered at $x_i$,
where
\[ \rho(\ell) = 8\epsilon^{-1}2^{\ell-\Lambda}\sqrt{d} . \]
We say that $x_i$ is~\emph{$(\epsilon,\Lambda)$-padded} in the quadtree $T$,
if it is $(\epsilon,\Lambda,\ell)$-padded for every level $\ell$ of $T$.
\end{definition}

Note that $d$ and $\Lambda$ are fixed parameters for a given input.
We omit their dependence from the notation $\rho(\ell)$ for simplicity.

We now prove~\cref{thm:ann}.
It follows directly from combining the following two lemmas.
\begin{lemma}\label{lmm:padding}
If the grids are randomly shifted, as in~\Cref{s:compression},
then every point $x_i$ is $(\epsilon,\Lambda)$-padded in $T$ with probability $1-\delta$.
\end{lemma}
\begin{proof}
Fix a point $x_i$, a coordinate $k\in\{1,\ldots,d\}$
and a level $\ell$.
Let $x_i(k)$ denote the value of $x_i$ in coordinate $k$.
Along this coordinate, we are randomly shifting a $1$-dimensional grid
partitioned into intervals of length $2^\ell$.
Since the shift is uniformly random,
the probability for $x_i(k)$ to be at distance at most
$\rho(\ell)$ from an endpoint of the interval that contains it equals
$2\rho(\ell)/2^\ell$. By plugging our setting of $\rho(\ell)$ and $\Lambda$,
this probability equals $\delta/(d\log\Phi)$.
Taking a union bound over the $d$ coordinates, we have probability at most $\delta/\log\Phi$
for $x_i$ to be at distance at most $\rho(\ell)$ from the boundary of the cell of $G_\ell$ that contains it.
In the complement event $x_i$ is $(\epsilon,\Lambda,\ell)$-padded in $G_\ell$.
Taking another union bound over the $\log\Phi$ levels in the quadtree,
$x_i$ is $(\epsilon,\Lambda)$-padded with probability at least $1-\delta$.
\end{proof}

\begin{lemma}\label{lmm:distances}
If a point $x_i$ is $(\epsilon,\Lambda)$-padded in $T$, then for every $j\in[n]$,
\[
  (1-\epsilon)\norm{\tilde x_i - \tilde x_j} \leq
  \norm{x_i-x_j} \leq
  (1+\epsilon)\norm{\tilde x_i - \tilde x_j},
\]
where $\{\tilde x_i\}$ are as defined in~\Cref{s:compression}.
\end{lemma}

\begin{proof}
We recall that $T$ is a pruned quadtree in which every node $v$ is associated with a grid cell
of an axis-parallel grid $G_{\ell}$ with side length $2^{\ell}$, which is aligned with and contained in $H$.
We call $\ell$ the~\emph{level} of $v$, and denote it henceforth by $\ell(v)$.
We will use the term ``bottom-left corner'' of a grid cell for the corner that minimizes
all coordinate values (i.e.,~the high-dimensional analog of a bottom-left corner in the plane).

Let $r$ be the root of $T$. We may assume w.l.o.g.~that the bottom-left corner
of $H$ is the origin in $\R^d$, since translating $H$ together with the entire pointset
does not change pairwise distances.
Under this assumption, we make the following observation, illustrated in~\Cref{fig:corners}.

\begin{observation}\label{obs:corners}
Let $v$ be a node in $T$. If the path from $r$ to $v$ contains only short edges,
then $c(v)$ 
(defined by the decompression algorithm in~\Cref{s:compression})
is the bottom-left corner of the grid cell associated with $v$.
\end{observation}

\begin{figure}[ht]
\vskip 0.2in
\begin{center}
\centerline{\includegraphics[scale=0.6]{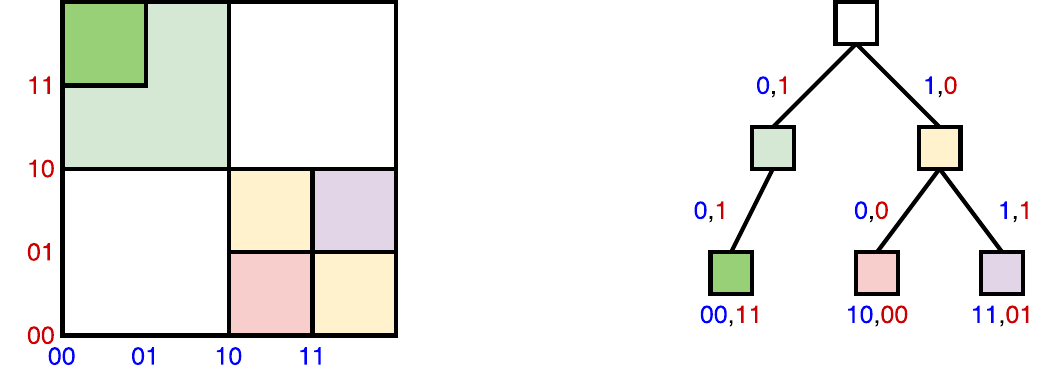}}
\caption{By collecting the edge label bits along every dimension from the root to a
node, and padding with zeros as necessary, we obtain the binary expansion of the
bottom-left corner of the associated grid cell.}
\label{fig:corners}
\end{center}
\vskip -0.2in
\end{figure} 

Let $x_i$ be a padded point, and $x_j$ be any point.
Recall that we denote by $v_i$ and $v_j$ the leaves corresponding to $x_i$ and $x_j$ respectively
(see~\Cref{s:compression}).
Let $w$ be the lowest common ancestor of $v_i$ and $v_j$ in $T$.
Since $x_i$ and $x_j$ are in separate grid cells of $G_{\ell(w)-1}$, and the cell
containing $x_i$ also contains the ball of radius $\rho(\ell(w)-1)$ around $x_i$,
we have
\begin{equation}\label{eq:separation}
  \norm{x_i-x_j} \geq \rho(\ell(w)-1) = 8\epsilon^{-1}2^{\ell(w)-1-\Lambda}\sqrt{d} .
\end{equation}

Let $u_i$ be the lowest node on the downward path from $w$ to $v_i$, that can be
reached without traversing a long edge.
Similarly define $u_j$ for $v_j$.
See~\Cref{fig:proof_picture} for illustration.

Note that $u_i$ must be either the leaf $v_i$, or an internal node whose
only outgoing edge is a long edge. In both cases, $u_i$ is the bottom of
a path of degree-$1$ nodes of length $\Lambda$:
\begin{itemize}
  \item If $u_i$ is a leaf: Since the pointset has aspect ratio $\Phi$, then after $\log\Phi$
  levels the grid becomes sufficiently fine such that each grid cell contains at most one point $x_i$.
  Since we generate the quadtree with $L=\log\Phi+\Lambda$ levels, then each point $x_i$ is in its
  own grid cell for at least the bottom $\Lambda$ levels of the quadtree.
  \item If $u_i$ is an internal node which is the head of a long edge: Since the pruning step only places
  long edges at the bottom of degree-$1$ paths of length $\Lambda$, then $u_i$ must be the bottom
  node of such path.
\end{itemize}

On the other hand $w$ is an ancestor of $u_i$, and it has degree at least $2$,
since it is also an ancestor of $u_j$. Hence $w$ is at least $\Lambda$ levels above $u_i$,
implying $\ell(v_i)\leq\ell(w)-\Lambda$.
Applying the same arguments to $u_j$ we get also $\ell(v_j)\leq\ell(w)-\Lambda$.

\begin{figure}[ht]
\vskip 0.2in
\begin{center}
\centerline{\includegraphics[width=0.4\textwidth]{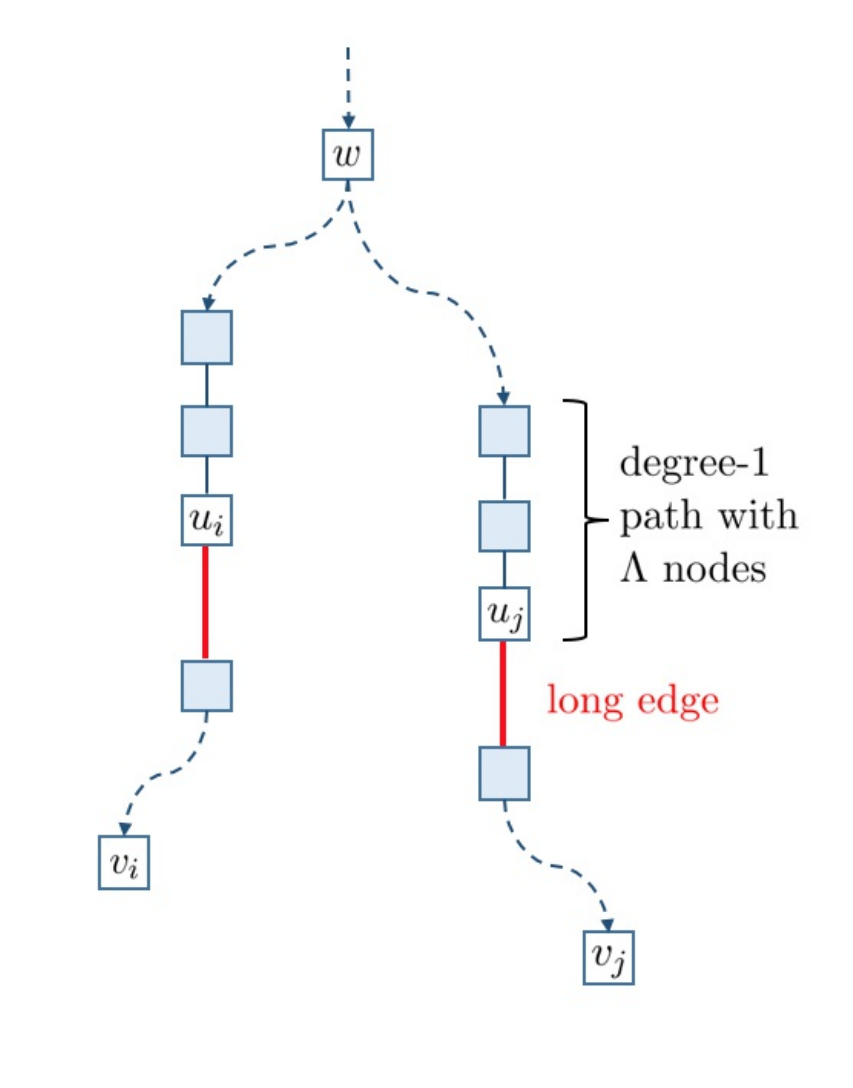}}
\caption{In the proof of~\cref{lmm:distances}, $w$ is the lowest common ancestor of $v_i,v_j$,
the leaves corresponding to $x_i,x_j$.
$u_i$ is the lowest node on the downward path from $w$ to $v_i$ which is achievable without
traversing any long edges (marked in red). $u_j$ is defined similarly for $v_j$.}
\label{fig:proof_picture}
\end{center}
\vskip -0.2in
\end{figure}

Let $c^*(u_i),c^*(u_j)\in\R^d$ be the bottom-left corners of the grid cells associated with $u_i$ and $u_j$.
If all edges on the downward paths from the root of $T$ to $u_i$ and $u_j$
were short, then~\cref{obs:corners} would yield that $c^*(u_i)=c(u_i)$ and  $c^*(u_j)=c(u_j)$.
In general, there might be some long edges on those paths, but they all must lie on the subpath
from the root of $T$ down to $w$, which is the same for both paths. This is because by the choice of
$u_i$ and $u_j$, all downward edges from $w$  to either of them are short.
Therefore $c(u_i)$ and $c(u_j)$
are shifted from the true bottom-left corners by the same shift, which we denote by
\[ \eta = c^*(u_i)-c(u_i) = c^*(u_j)-c(u_j) . \]

Next, observe that the grid cell associated with $u_i$ has
side $2^{\ell(u_i)}$ and it contains both $c^*(u_i)$ and $x_j$.
Therefore $\norm{x_i-c^*(u_i)} \leq 2^{\ell(u_i)}\sqrt{d}$.

Furthermore, since $u_i$ is an ancestor of $v_i$, then by the definition of $c(u_i)$ and $c(v_i)$,
in each coordinate, the binary expansions of these two vertices are equal from the location $\ell(u_i)$ and up.
In the less significant locations, $c(u_i)$ is zeroed while $c(v_i)$ may have arbitrary bits.
This means that the difference between $c(u_i)$ and $c(v_i)$ in each coordinate can be at most
$2^{\ell(u_i)}$ in the absolute value, and consequently $\norm{c(v_i)-c(u_i)} \leq 2^{\ell(u_i)}\sqrt{d}$.
Recalling that the decompression algorithm defines $\tilde x_i=c(v_i)$, we get
$\norm{\tilde x_i-c(u_i)} \leq 2^{\ell(u_i)}\sqrt{d}$.

Collecting the above inequalities, we have
\begin{align*}
  \norm{x_i-\eta-\tilde x_i} &= \norm{x_i-c(u_i)-\eta+c(u_i)-\tilde x_i} \\
  &=  \norm{x_i-c^*(u_i)+c(u_i)-\tilde x_i} \\
  &\leq \norm{x_i-c^*(u_i)} + \norm{c(u_i)-\tilde x_i} \\
  &\leq 2\cdot2^{\ell(u_i)}\sqrt{d} \\
  &\leq 2\cdot2^{\ell(w)-\Lambda}\sqrt{d}.
\end{align*}
Similarly for $j$ we have $\norm{x_i-\eta-\tilde x_i} \leq 2\cdot2^{\ell(w)-\Lambda}\sqrt{d}$.
Together, by the triangle inequality,
\begin{align*}
  \norm{\tilde x_i - \tilde x_j} &= \norm{\tilde x_i +\eta - x_i + x_i - x_j + x_j -\eta - \tilde x_j} \\
  &= \norm{x_i-x_j} \pm \left( \norm{x_i-\eta-\tilde x_i} + \norm{x_i-\eta-\tilde x_i} \right) \\
  &= \norm{x_i-x_j} \pm 4\cdot2^{\ell(w)-\Lambda}\sqrt{d}. \\
\end{align*}
To complete the proof of~\cref{lmm:distances} it remains to show
$4\cdot2^{\ell(w)-\Lambda}\sqrt{d} \leq \epsilon\cdot\norm{x_i-x_j}$,
which follows from \Cref{eq:separation}.
\end{proof}

\subsection{Sketch Size and Running Time}
\begin{lemma}\label{lmm:sketch_size}
QuadSketch produces a sketch of size $O(nd\Lambda+n\log n)$ bits.
\end{lemma}
\begin{proof}
The tree $T$ has $n$ leaves, and we have pruned each non-branching path in it
to length $\Lambda$. Hence its total size is $O(n\Lambda)$, and its structure
can be stored with this many bits using (for example) the DFS scan described in~\Cref{s:compression}.
Each short edge label is $d$ bits long, so together they consume $O(nd\Lambda)$ bits.
As for the long edges, there can be at most $O(n)$ of them, since the bottom of
each long edge is either a branching node or a leaf. The long edge labels are lengths
of downward paths in the non-pruned tree $T^*$, whose height bounded by is $O(\log\Phi+\Lambda)$.
Together the long edge labels consume $O(n\log(\log\Phi+\Lambda))$ bits, which is dominated by $O(n\Lambda)$.
Finally for each point $x_i$ we store the index of its corresponding leaf $v_i$, and since there are
$n$ leaves, this requires $O(n\log n)$ additional bits to store.
\end{proof}

\begin{lemma}\label{lmm:runtime}
The QuadSketch construction algorithm runs in time $O(ndL)$.
\end{lemma}
\begin{proof}
Given a quadtree cell and a point contained in it, in order to bucket the point into a cell
in the next level, we need to check for each coordinate whether the point falls in the upper
or lower half of the cell. This takes time $O(d)$. Since each point is bucketed once in every
level, and we generate $T^*$ for $L$ levels, the quadtree construction time is $O(ndL)$.
The pruning step requires just a linear scan of $T^*$, in time $O(nL)$.
\end{proof}

\subsection{Maximum Distortion}
We now prove~\Cref{thm:max_distortion}.

\paragraph{Sketching algorithm}
Given a pointset $X$, apply QuadSketch to $X$ and let $T_1$ be the resulting tree.
Let $Q\subset X$ be the padded points in $T_1$
(meaning those for which the condition of~\cref{lmm:padding} is satisfied for $T_1$).
Continue by recursion on $X\setminus Q$,
until all points in $X$ are padded in some tree. The returned sketch contains all trees $T_1,\ldots,T_k$
constructed during the recursion, and in addition, for every point $x_i$ we store the index $\gamma_i$
of the tree in which it is padded.

\paragraph{Query algorithm}
Given two point indices $i,j$, assume w.l.o.g.~$\gamma(i)\leq\gamma(j)$, then the tree $T_{\gamma(i)}$
has corresponding leaves for both $x_i$ and $x_j$. We decompress $\tilde x_i$ and $\tilde x_j$
from $T_{\gamma(i)}$ and return $\norm{\tilde x_i-\tilde x_j}$.

\paragraph{Analysis}
The correctness of the estimate up to distortion $1\pm\epsilon$ follows from~\cref{lmm:distances}.
We now bound the sketch size and the running time. \Cref{lmm:padding} with $\delta=0.25$ implies
that in each of the trees $T_1,\ldots,T_k$, the expected fraction of padded points is $0.75$.
Hence by Markov's inequality, with probability $0.5$ at least half the points are padded.
Since the calls to QuadSketch are independent (and its success probability depends only on its
internal randomness and not on the input points), with probability $0.5^{\lceil\log_2n\rceil}\sim1/n$ this happens
in each of the first $k=\lceil\log_2n\rceil$ iterations. This probability can be amplified to constant by $O(\log n)$
independent repetitions. If this event has happened then the sketching algorithm terminates since $Q$
becomes empty. Therefore the total running time of a successful execution
is $O(\log^2n)$ calls to QuadSketch, which by~\cref{lmm:runtime} is  $\tilde O(ndL)$.

Furthermore, since the number of padded points decreases by at least half in every iteration, the total size
of the sketches $T_1,\ldots,T_k$ is
\[ O\left(\sum_{k'=0}^{k-1}\frac{n}{2^{k'}}(d\Lambda+\log \frac{n}{2^{k'}})\right)=O(n(d\Lambda+\log n)) , \]
the same as in~\cref{thm:ann} up to a constant factor. Finally, since each $\gamma(i)$ is index in
$\{1,\ldots,\lceil\log_2n\rceil\}$, the $\gamma(i)$'s only take additional $O(n\log\log n)$ bits to store. \qed

\section*{Acknowledgements}
 This research was supported by grants from NSF, Simons Foundation, and Shell (via MIT Energy Initiative).

\bibliography{example_paper}
\bibliographystyle{amsalpha}

\end{document}